\def\year{2021}\relax
\documentclass[letterpaper]{article} % DO NOT CHANGE THIS
\usepackage{aaai21}  % DO NOT CHANGE THIS
\usepackage{times}  % DO NOT CHANGE THIS
\usepackage{helvet} % DO NOT CHANGE THIS
\usepackage{courier}  % DO NOT CHANGE THIS
\usepackage[hyphens]{url}  % DO NOT CHANGE THIS
\usepackage{graphicx} % DO NOT CHANGE THIS
\usepackage{graphicx}  % DO NOT CHANGE THIS
\usepackage{natbib}  % DO NOT CHANGE THIS AND DO NOT ADD ANY OPTIONS TO IT
\usepackage{caption} % DO NOT CHANGE THIS AND DO NOT ADD ANY OPTIONS TO IT
\usepackage{amsmath}
\usepackage{mathrsfs}
\usepackage{framed}
\usepackage{tikz}
\usepackage{subfigure}
\usepackage{enumerate}
\usepackage{amsfonts}
\usepackage{algorithm}
\usepackage{algpseudocode}
\usepackage{amsmath}
\usepackage{graphics}
\usepackage{epsfig}
\usepackage{multirow} 
\usepackage{amsthm}
\newtheorem{theorem}{Theorem}
\newtheorem{prop}{Proposition}
\newtheorem{Def}{Definition}
\title{Barter Exchange via Friends' Friends}
\author {
        Yue Zheng,
        Tianyi Yang,
        Wen Zhang and
        Dengji Zhao\\
}
\affiliations {
     School of Information Science and Technology, ShanghaiTech University, Shanghai, China \\
    \{zhengyue1, yangty, zhangwen, zhaodj\}@shanghaitech.edu.cn
}

\begin{document}

\maketitle

\begin{abstract}
Barter exchange studies the setting where each agent owns a good, and they can exchange with each other if that gives them more preferred goods. This exchange will give better outcomes if there are more participants. The challenge here is how to get more participants and our goal is to incentivize the existing participants to invite new participants. However, new participants might be competitors for the existing participants. Therefore, we design an exchange mechanism based on the classical Top Trading Cycle (TTC) algorithm to solve their conflicts. Our mechanism is truthful in terms of revealing their preferences and also guarantees that inviting all their neighbors is a dominant strategy for all participants. The mechanism can be applied in settings where more participants are preferred but no extra budget to reach new participants. 
\end{abstract}

\noindent 
\section{Introduction}
Recently more and more people enjoy the convenience of online shopping, while at the same time, how to deal with the goods they do not need is also a problem for them. There are a variety of ways to deal with these goods, such as donating, selling, or throwing them away. Sometimes people would set up stalls outside their house to sell these unwanted goods to neighbors, but it is hard to settle down a price that they are both satisfied with. Barter exchange is an old-fashioned way to exchange goods, where the money is not used as a medium of the exchange. It is a simple and eco-friendly way for people to deal with unwanted goods. Therefore, barter exchange seems more attractive to people. 

Barter exchange has already grown up in a modern way. Since the Internet can bring traders together, many online bartering websites and communities are growing up. Nowadays, people can find many reliable and large-scale bartering platforms to swap their clothes, services, and even their houses. However, the choices for users are limited and the platforms restrict the swaps between pairs of users. For example, GoSwap.org is a website where an individual can swap her house with another user. However, user $a$ might fail to swap with $b$ when $b$ prefers another individual $c$'s house, where $c$ is their common friend but she has not registered for the website. Even if $c$ prefers $a$'s house, no pairs among them can successfully swap. However, if $c$ is invited to join the website and a trade cycle can be formed, then $a,b,c$ can each get their most favorite house. 

Therefore, in this paper, we propose a novel matching method for an organizer to hold a larger barter exchange. We model all the agents (individuals) as a large scale social network where each agent is linked with some other agents (known as neighbors). The organizer initially invites some agents to join the bartering market. For agents who have already been in the market, they are incentivized to further invite their neighbors and so on. Our mechanism guarantees that each agent will get a good that is at least as good as the one when not inviting neighbors. In the end, all agents in the social network will be in the market and may get a better exchange. 

The pattern of barter exchange is in the class of one-sided matching problems, which has been well studied in the traditional settings. Normally, agents in the one-sided market need to be matched with goods~\cite{abdulkadiroglu2013matching}. These goods in fact do not have preferences over the agents or the allocation. A lot of related theoretical work are at the interface of computer science, economics, and game theory~\cite{haeringer2018market}. In addition, most of the work captured many real-world problems, such as the kidney exchange problem~\cite{roth2004kidney,sonmez2020incentivized}, the assignment problem~\cite{gardenfors1973assignment,andersson1998contract}, the house allocation problem~\cite{hylland1979efficient} and the housing market~\cite{shapley1974cores,anno2015short}. Our mechanism is similar to the setting of the housing market, where every agent initially owns a good and wants to exchange her good with others such that she can receive a better one.

~\citeauthor{shapley1974cores}~\citeyear{shapley1974cores} firstly described the Top Trading Cycle (TTC) algorithm in the housing market, which is a classical and elegant algorithm with many desirable properties. The TTC algorithm has a great influence on one-sided matching literature. Most of these work focused on the game theoretical analysis and assume that agents are able to communicate with each other and thus, they can directly trade~\cite{roth2005pairwise,aziz2012housing,damamme2015power}. However, this assumption fails to cope with the reality that some agents are unable to contact with each other, especially in a large market. Therefore, it is more natural to assume that agents form a network.
\iffalse
In addition, there are many remarkable works and real-life applications basing on the TTC algorithm. Roth and Postlewaite pointed out that the allocation outputted by the TTC algorithm is an unique core allocation in~\cite{roth1977weak}. Most of these works focusing on the game theoretical analysis and assume that agents are able to mutually communicate and thus, they can directly trade. However, this assumption fails to cope with the reality that some agents are unable to contact each other, especially in a large scale market. 
\fi

Though there have been some research studying matching in social networks, they mainly focused on two-sided matching problems.~\citeauthor{jackson2010social}~\citeyear{jackson2010social}, ~\citeauthor{calvo2007networks}~\citeyear{calvo2007networks}, and~\citeauthor{calvo2004effects}~\citeyear{calvo2004effects} analyzed the factors of the network that would affect labors' staying or leaving the market. Similarly, ~\citeauthor{arcaute2009social}~\citeyear{arcaute2009social} studied job market problems with the constraint that workers learn only about job information through social contacts. They show that a simple variation of the Gale-Shapley mechanism converges to a locally stable solution. ~\citeauthor{gourves2017object} ~\citeyear{gourves2017object} studied the housing market in a social network and focused on the trade between pairs of neighbors. Moreover,~\citeauthor{bailey2018economic}~\citeyear{bailey2018economic} studied the effects of social interactions on individuals’ housing market expectations. The purpose of considering social network in these work is to achieve a more realistic model which considers the connections of agents. However, the networks in their settings are static and given in advance. In this paper, we do not only consider the connections of agents but also utilize the propagation via the social network to attract more agents to participate in the exchange. This idea is inspired by the work of~\citeauthor{li2017mechanism}~\citeyear{li2017mechanism} and~\citeauthor{zhao2019selling}~\citeyear{zhao2019selling}. They studied selling items via social networks and proposed truthful mechanisms to incentivize participants to help a seller propagate the sale. Our work, to the best of our knowledge, is the first work to study one-sided matching problem in a social network, where the social network is not given in advance but is generated by the agents' propagation. 

\section{The Model}
We consider a barter exchange problem in a social network represented by an indicted graph $G=(V,E)$, which involves a market organizer $o$ and $n$ agents (nodes) denoted by $N=\{1,2,\cdots,n\}$, i.e., $V=N\cup\{o\}$. For each agent $i\in N$, she is endowed with a good $g_i$ ready for an exchange. Let $\mathcal{G}=\{g_1,\cdots,g_n\}$ be the set of all goods brought by the agents. We also use $i$ to indicate the owner of the good $g_i$. For each $i\in V$, she has a set of neighbors $r_i\subseteq V$. That is for any $j\in r_i$, there is an edge $e(i,j)$ representing that $i$ can directly exchange information with $j$.

Traditionally, the organizer can find an allocation among his friends by using the Top Trading Cycle (TTC) Algorithm, which satisfies desirable properties. However, without any further propagation, agent $i$ can only exchange her good with neighbors as they are not aware of the rest of the network. Therefore, there exists a limited number of participants (i.e., a limited number of goods) in the traditional setting. The goal of the organizer is to design a mechanism that incentivizes agents to invite their neighbors to join the barter exchange so that they can get better exchanges. In order to achieve a better allocation, it would be better if all agents in the network could join the market. Nevertheless, agents might be unwilling to invite neighbors as they may compete for the same good. Thus, the exchange mechanism should guarantee that each agent will get a good that is at least as good as the one she can get when not participating in. 

Let us first formally describe the model. For each $i\in N$, $i$ has a preference function $\mathcal{P}_i:\mathcal{G}'\rightarrow \succ_i$, for $\forall\ \mathcal{G}' \subseteq \mathcal{G}$. The preference function outputs a strict, complete and transitive preference over the goods in $\mathcal{G}'$, denoted by a linear order $\succ_i$, where $a\succ_i b$ means that $i$ prefers good $a$ to good $b$. Let $\theta_i=(\mathcal{P}_i,r_i)$ be the type of agent $i\in N$, $\theta=(\theta_1,\cdots,\theta_n)$ be the type profile of all agents and $\theta_{-i}$ be the type profile of all agents except $i$. $\theta$ can be also written as $(\theta_i,\theta_{-i})$. Let $\Theta_i$ be the type space of agent $i$ and $\Theta$ be the type profile space of all agents. 

The exchange mechanism requires that each agent, who participates in the bartering market, to invite all their neighbors and to report her preference function. In this paper, an agent $i$'s action in the mechanism is defined as $\theta_i'=(\mathcal{P}_i',r_i')$, where $\mathcal{P}_i'$ is the preference function reported by agent $i$ and $r_i'\subseteq r_i$ is the set of neighbors that $i$ actually has invited. The action profile of all agents is denoted as $\theta'=(\theta_1',\cdots,\theta_n')$. Let $\theta_j'=nil$ if $j$ has not been informed about the bartering exchange. 

Let $G(\theta')=(V(\theta'),E(\theta'))$ be the graph generated by the reported type profile $\theta'$, where $V(\theta')\subseteq N\cup\{o\}$. The construction is generated as follows:
\[\begin{cases}
  i\in V(\theta')& \text{if $i=o$ or $i\in r_j'$ where $j\in V(\theta')$}\\
  e(i,j)\in E(\theta')&\text{if $i\in r_j'$, $j\in r_i'$ and $i,j \in V(\theta')$}
\end{cases}\]
In fact, the reported type may not be the same as the true type. Therefore, $i\notin V(\theta')$ means that $i$ is not invited and hence, $\theta_i'=nil$. Figure~\ref{1a} shows a simple social network and Figure~\ref{1b} is one corresponding generated graph $G(\theta')$. In Figure 1, $o$ is the organizer and other numbered nodes are agents. Suppose $\theta'$ is the reported type profile where all agents report their true type except that $2$ misreported $\theta_2'=(\mathcal{P}_2',\emptyset)$, i.e., agent $2$ does not invite her neighbors. Since agent $6,7,9$ cannot be in the market without $2$'s invitation, we set their reported type to be $nil$.
\begin{figure}[htbp]%
	\centering
	\subfigure[]{
		\label{1a}
		\includegraphics[width=0.26\linewidth]{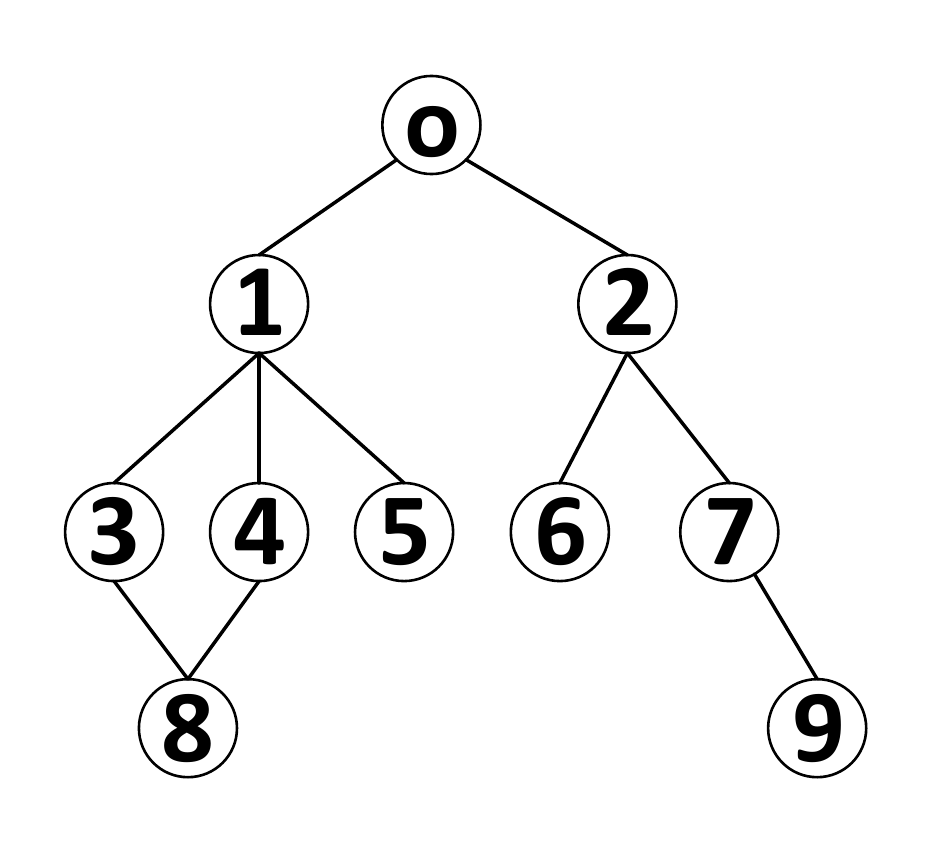}}
	\hspace{20pt}
	\subfigure[]{%
		\label{1b}%
		\includegraphics[width=0.26\linewidth]{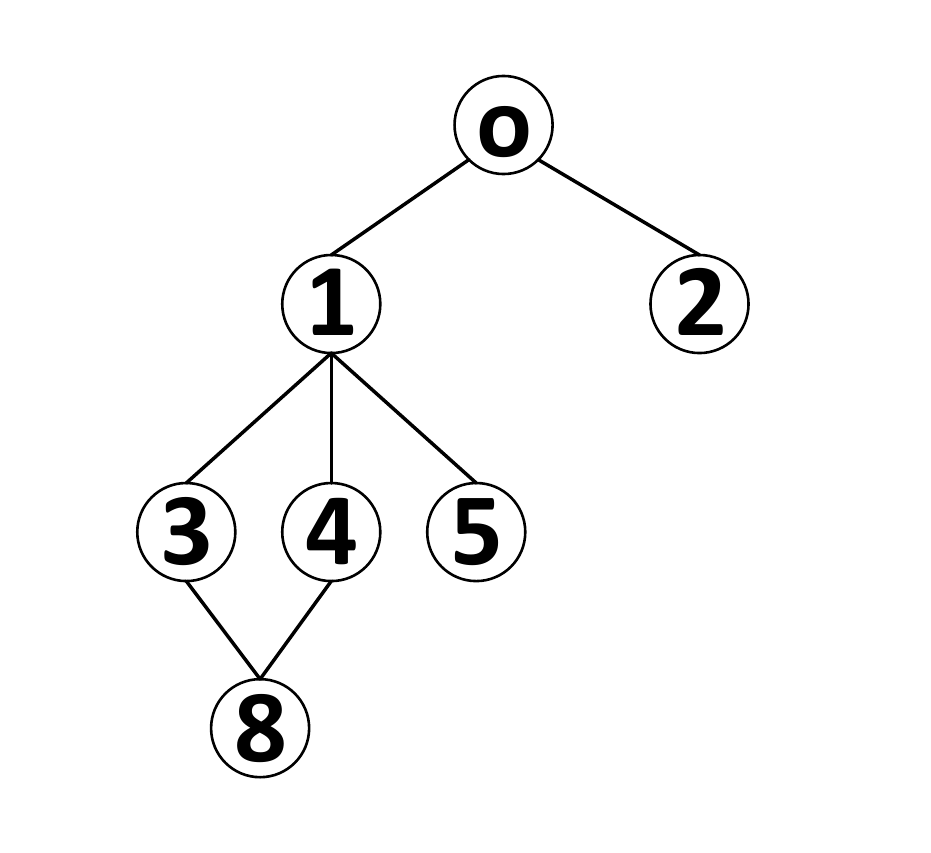}}\\
	\caption{An example of a social network $G$ and its generated graph $G(\theta')$ where agent $2$ did not invite agents $6,7$.}
\end{figure}

In the rest of the paper, we only focus on the generated graph. Given a generated graph $G$, we define some basic concepts as follows.
\begin{itemize}
    \item If all the paths from $o$ to $j$ have to pass $i$, we say $i$ is $j$'s ancestor and $j$ is $i$'s descendant.
    \item For each $i\in V\setminus \{o\}$, let $D_i$ be the set of descendants of $i$ ($i$ is included). Let $A_i$ be the set of ancestors of $i$. 
    \item Let $G_i=(V_i,E_i)$ be the subgraph induced by $D_i$, where $V_i=D_i$ and $e(j,k)\in E_i$ if $e(j,k) \in E$ and $j,k\in D_i$. Note that $i$ is the ancestor for all nodes in $D_i$. 
    %We use $G_{-D_i}$ denote the graph without the participation of $i$ and all other the agents in $D_i$.
\end{itemize}

A one-sided matching mechanism produces an allocation rule which gives a matching (allocation). Here, we give some basic and important definitions of the matching. 
\begin{Def}
\rm (\textit{Matching}) Given a barter exchange problem, a matching mechanism outputs a \textit{matching} $\mu_{G,\theta'}: N\rightarrow \mathcal{G}$, which maps each good in $\mathcal{G}$ to an agent in $N$ under the reported type profile $\theta'$ and the generated graph $G$, where $\mu_{G,\theta'}(i)$ denotes the good allocated to agent $i\in N$ in $\mu_{G,\theta'}$.
\end{Def}
%Let $\mu_0$ denotes the initial allocation, that is $\mu_0(i)=g_i$ for each agent $i\in V\setminus \{o\}$.
\begin{Def}
\rm (\textit{Pareto Optimal}) Given a barter exchange problem and the generated graph $G=(V,E)$, we say a matching $\mu$ is Pareto dominated by some other matching $\nu$ if for each $i\in V\setminus \{o\}$:
$\nu(i)\succeq_i \mu(i)$, and for some $j\in V\setminus \{o\}$: $\nu(j)\succ_i \mu(j)$. That is all agents in $V\setminus \{o\}$ weakly prefer $\nu$ over $\mu$, and at least one agent strictly prefers $\nu$. We say $\mu$ is \textit{Pareto Optimal} if no other matching Pareto dominates $\mu$.
\end{Def}

\begin{Def}
\rm (\textit{Individual Rationality}) A matching mechanism is \textit{individually rational} (IR) if for every agent $i\in N$, all $\theta', \theta''\in \Theta$, we have $\mu_{G,\theta'}(i)\succeq_i \mu_{G_i,\theta''}(i)$ where $\theta_i'=(\mathcal{P}_i,r_i')$ and $\theta''$ is the corresponding reported profile of the agents in $D_i$ such that $\theta_j''=nil$ for all $j\notin D_i$ and others are the same as those in $\theta'$. 
\end{Def}

\iffalse
\begin{Def}
\rm A matching mechanism $\mathcal{M}$ is \textit{individually rational} (IR) if for every agent $i\in V\setminus \{o\}$, all $\theta', \theta''\in \Theta$, and $\mathcal{M}(\theta')=\mu_{G,\theta'}$, we have $\mu_{G,\theta'}(i)\succeq_i h_j$ where $h_j$ is the good $i$ can get from her neighbor $j$ by an one-to-one exchange. 
\end{Def}
\fi
Different from the traditional definition of IR, here IR indicates that the good agent $i$ gets from the market is at least as good as the one (according to her true preferences) she gets when she exchanges in local (among her descendants), no matter what the others do. Moreover, in our setting, we want to incentivize agents to not only just report their preferences truthfully, but also invite all their neighbors to join the bartering market. Therefore, we extend the traditional definition of incentive compatibility to cover the invitation part.
\begin{Def}
\rm (\textit{Incentive Compatibility}) We say a matching mechanism is \textit{incentive compatible} (IC) if for all $i\in V\setminus \{o\}$ and all $\theta', \theta''\in \Theta$, we have $\mu_{G,\theta'}(i)\succeq_i \mu_{G',\theta''}(i)$, where $\theta'=(\theta_i,\theta_{-i}')$ and $\theta''=(\theta_i',\theta_{-i}')$ and $G'$ is the generated graph under the reported type profile $\theta''$.
\end{Def}

In this paper, we design a one-sided matching mechanism that is IC, IR, and outputs a Pareto optimal matching. 
%for agents to invite all their neighbors to the bartering market, so that they can potentially get better exchanges. 

\section{TTC in Social Networks}
Given the generated graph $G$, we first consider whether the Top Trading Cycle algorithm can be directly extended to the network setting without sacrificing its properties. %We outline the Top Trading Cycle algorithm as follows. 
\begin{framed}
	\textbf{The Top Trading Cycle (TTC) Algorithm}\\
	\noindent\rule{\textwidth}{0.35mm}
	Given the reported type profile $\theta'$, generate the graph $G=(V,E)$. 
    \begin{enumerate}
    \item Initially set $N'=V\setminus \{o\}$.
    \item While $N'\neq \emptyset$:
    \begin{enumerate}
        \item Construct a directed preference graph $G'=(V',E')$, where $V'=N'$ and for each $i,j\in N'$, the direct edge $e'(i,j)\in E'$ if and only if $g_j\succ_i g_k$ for all other $k\in V'$, i.e., each agent $i\in N'$ points to her most favorite remaining good (including $g_i$).
        \item Find an arbitrary directed cycle $c=(V_c,E_c)$ in $G'$. Let agents in the cycle trade directly, i.e., for every directed edge $e'(i,j)\in E_c$, set $\mu_{G,\theta'}(i)=g_j$. Remove all agents in the cycle from $N'$.
    \end{enumerate}
\end{enumerate}
\end{framed}
The TTC algorithm is simple and intuitive, but it is problematic if we directly apply the algorithm in the network setting.
\begin{prop}
\rm Given the generated graph $G$, the extended incentive compatibility and individual rationality, the TTC algorithm in social networks is neither IC nor IR.
\end{prop}
\begin{proof}
\rm We use an example to give the proof. Consider the setting given in Figure 2, where the nodes $o,1,2,3$ and the edges colored in black form a generated social network, and the blue edges represent agents' preferences. The preference order of each agent is listed as follows.
\[\begin{cases}
  g_3\succ_1 g_2\succ_1 g_1\\
  g_1\succ_2 g_3\succ_2 g_2\\
  g_1\succ_3 g_2\succ_3 g_3
\end{cases}\]
Here, if we directly apply the TTC algorithm, then agent $1$ and agent $3$ will exchange their goods, while agent $2$ leaves with her own good. However, as is shown in Figure~\ref{2b}, if agent $2$ does not invite $3$, then agent $1$ and agent $2$ can exchange their goods. Since agent $2$ gets a better good by not inviting her neighbor $3$, directly applying TTC in social networks is not IC. 

Consider the case in Figure~\ref{2c}, when $2$ does not participate in the market, then agent $2$ and agent $3$ exchange their goods, which is also better than the case in Figure~\ref{2a}. Thus, agent $2$ gets a worse good when participating in the market, i.e., TTC in the network is not IR. 
\end{proof}
\begin{figure}[htbp]%
	\centering
	\subfigure[]{
		\label{2a}
		\includegraphics[width=0.3\linewidth]{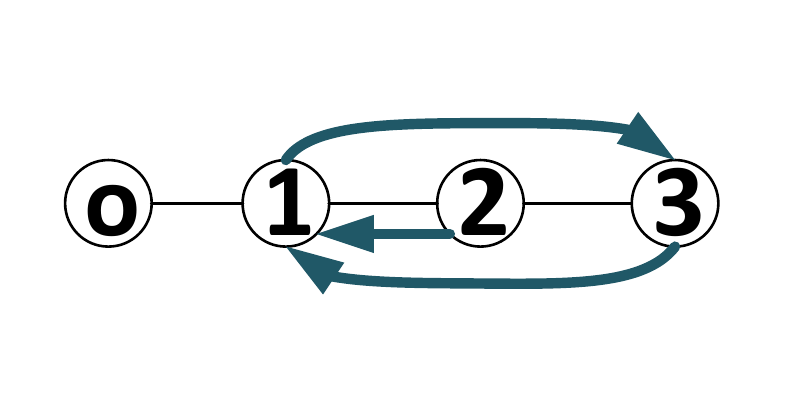}}
	\subfigure[]{%
		\label{2b}%
		\includegraphics[width=0.3\linewidth]{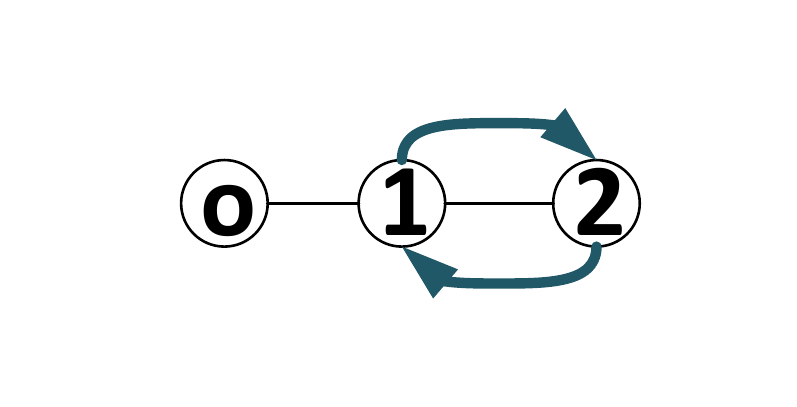}}
	\subfigure[]{%
		\label{2c}%
		\includegraphics[width=0.3\linewidth]{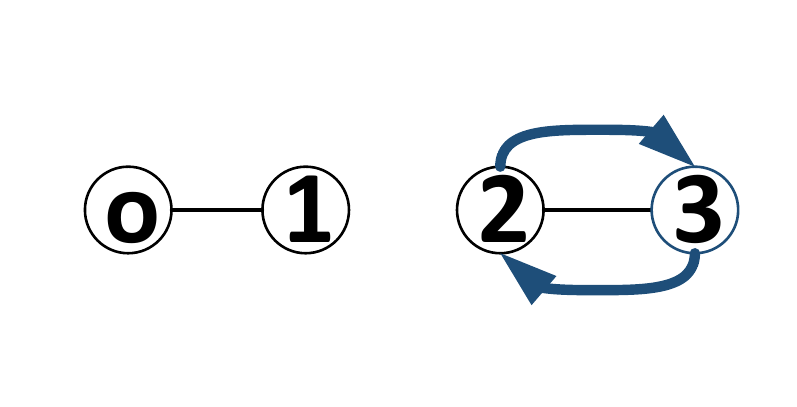}}\\
	\caption{A line-structured generated graph and the corresponding preference graph.}
\end{figure}
The example in Figure 2 reveals the problem that agents will not be incentivized to invite neighbors if that brings them competitors. Moreover, agents are not willing to participate in the mechanism since they probably get worse goods. 

The TTC algorithm in social networks fails to guarantee individual rationality and the invitation part of the incentive-compatibility. In this paper, we propose a novel mechanism which does not only incentivize agents to invite their neighbors but also promises that inviting neighbors can help them to get better goods.

\section{Impossibility Result}
In the bartering market, it would be nice if agents are allowed to point to all the goods in the market. It is more attractive for agents to participate in and further propagate the market, especially when the mechanism is intended to incentivize participants to invite more agents. However, it is hard to maintain all the desirable properties while giving the agents the freedom to choose. We prove that the property of incentive compatibility and Pareto optimal are incompatible when all the goods in the market are available for each agent.
\begin{table}[htbp]%
\centering
\begin{tabular}{|c|c|}
\hline
Agent $i$ & preference order $\succ_i$\\ 
\hline
1 & $g_4\succ_1 g_5\succ_1 g_1\succ_1 g_2\succ_1 g_3$ \\
\hline
2 &  $g_1\succ_2 g_2\succ_2 g_3\succ_2 g_4\succ_2 g_5$ \\
\hline
3 &  $g_4\succ_3 g_3\succ_3 g_1\succ_3 g_2\succ_3 g_5$ \\
\hline
4 & $g_3\succ_4 g_1\succ_4 g_4\succ_4 g_2\succ_4 g_5$ \\
\hline
5 &  $g_2\succ_5 g_5\succ_5 g_1\succ_5 g_3\succ_5 g_4$ \\
\hline
\end{tabular}
\caption{The preferences of agents in Figure 3.}
\end{table}
\begin{theorem}
\rm Given the barter exchange problem in a social network, no mechanism exists that is incentive compatible and outputs a Pareto optimal matching.
\end{theorem}
\begin{proof}
We prove Theorem 1 by an example given in Figure 3, where there exist 5 agents and their preference orders are shown in Table 1. The preference graph is also shown in Figure~\ref{3a}. For agent $1$, she can get a better good by being in the cycle involving agent $1,5,2$ or the cycle involving agent $1,4$. Similarly, agent $2$ and agent $5$ can only receive a better good by the trade cycle involving agent $1,5,2$. Agent $3$ hopes the cycle including $3$ and $4$ can trade. Agent $4$ can get the better exchange by two cycles, the cycle formed by agent $3,4$ and the cycle including $1,4$. There exist two Pareto optimal allocations which are shown in Figure~\ref{3b} and Figure~\ref{3c} respectively. This is because agent $1$ prefers the allocation in Figure~\ref{3c}, while agent $2$ prefers the allocation in Figure~\ref{3b}. We discuss these two allocations as follows.
\begin{itemize}
    \item Consider a mechanism that outputs the allocation as Figure~\ref{3b} shows. Note that agent $1$ prefers $g_4$ to $g_5$ and her descendant $3$ receives $g_4$ instead, since $3$ cannot be in the market without $1$'s invitation, then agent $1$ can deliberately not invite $3$. The corresponding graph is shown in Figure~\ref{4a}. If the cycle formed by agent $1,4$ can directly trade, then the mechanism does not satisfies the property of incentive compatibility since $1$ gets a better good by not inviting her neighbor $3$. However, if the mechanism does not allow this cycle to trade, then agent $1,2,3$ all leave with their own good, which is Pareto dominated by the allocation shown in Figure~\ref{4b} and hence, the matching that the mechanism outputs is not Pareto optimal.
    \item Consider another mechanism that outputs the allocation as Figure~\ref{3c} shows. If agent $2$ does not invite agent $4$ as shown in Figure~\ref{4c}, then the cycle involving agent $1,2,5$ can directly trade. Therefore, agent $2$ receives a better good by misreporting her type. Hence, this mechanism is not incentive compatible.
\end{itemize}
\end{proof}
\begin{figure}[htbp]%
	\centering
	\subfigure[]{
		\label{3a}
		\includegraphics[width=0.3\linewidth]{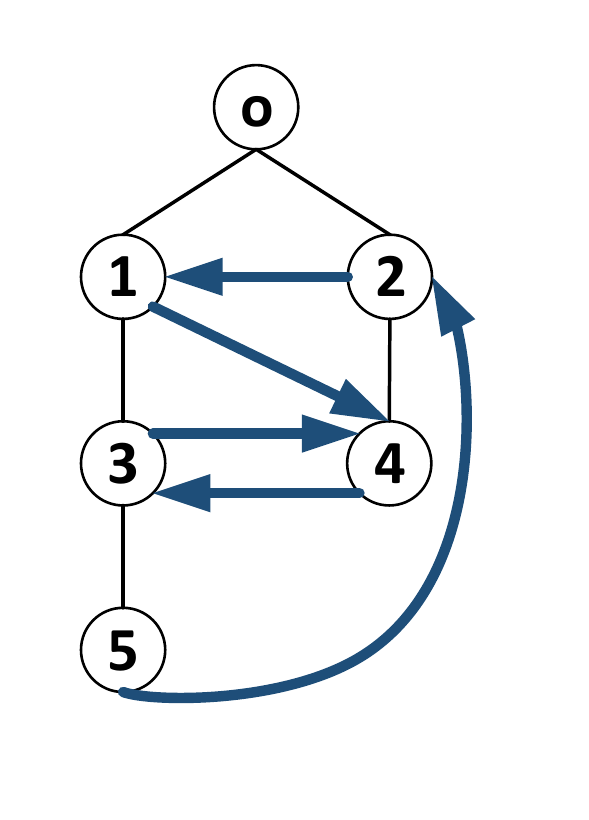}}
	\subfigure[]{%
		\label{3b}%
		\includegraphics[width=0.3\linewidth]{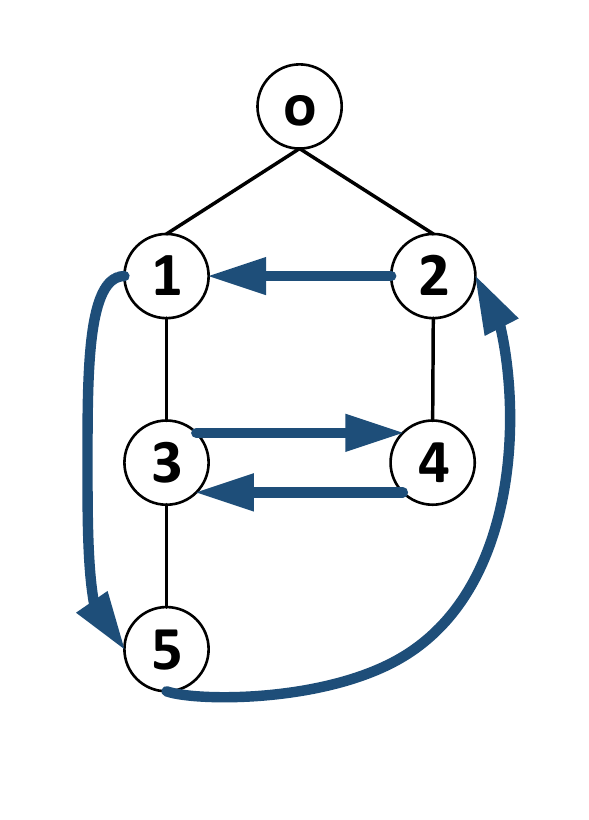}}
	\subfigure[]{%
		\label{3c}%
		\includegraphics[width=0.3\linewidth]{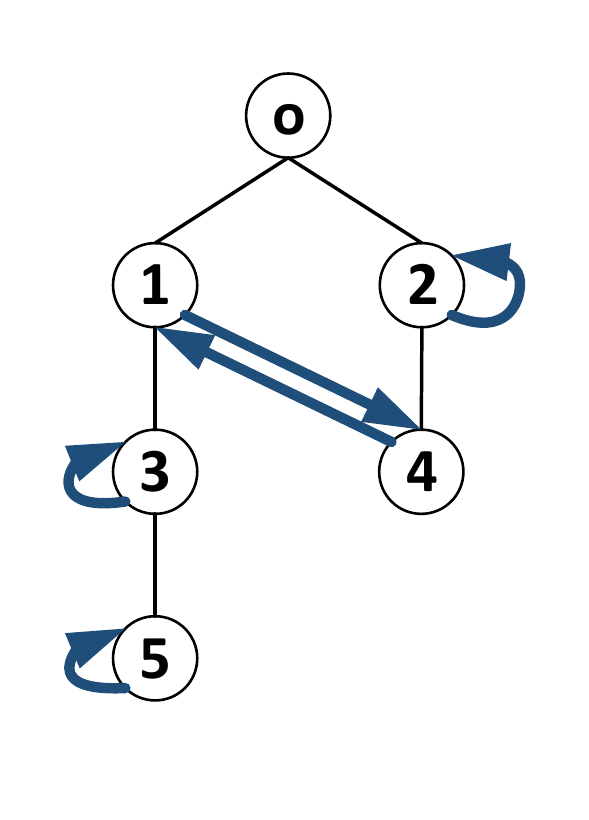}}\\
	\caption{The example to help prove theorem 1.}
\end{figure}
\begin{figure}[htbp]%
	\centering
	\subfigure[]{
		\label{4a}
		\includegraphics[width=0.3\linewidth]{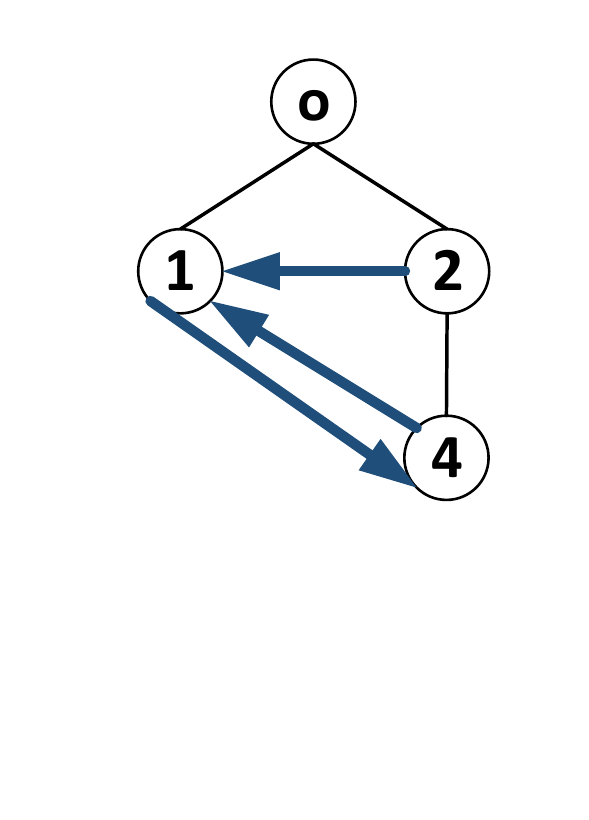}}
	\subfigure[]{%
		\label{4b}%
		\includegraphics[width=0.3\linewidth]{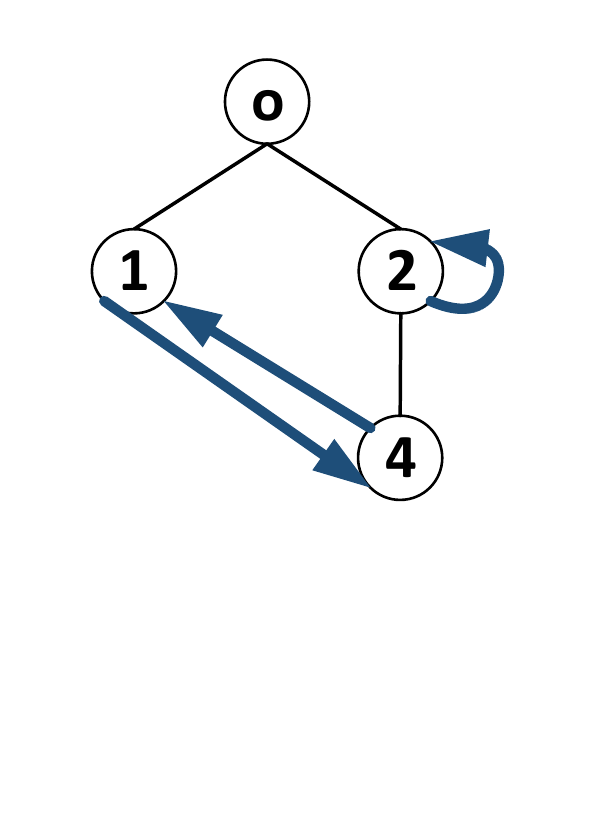}}
	\subfigure[]{%
		\label{4c}%
		\includegraphics[width=0.3\linewidth]{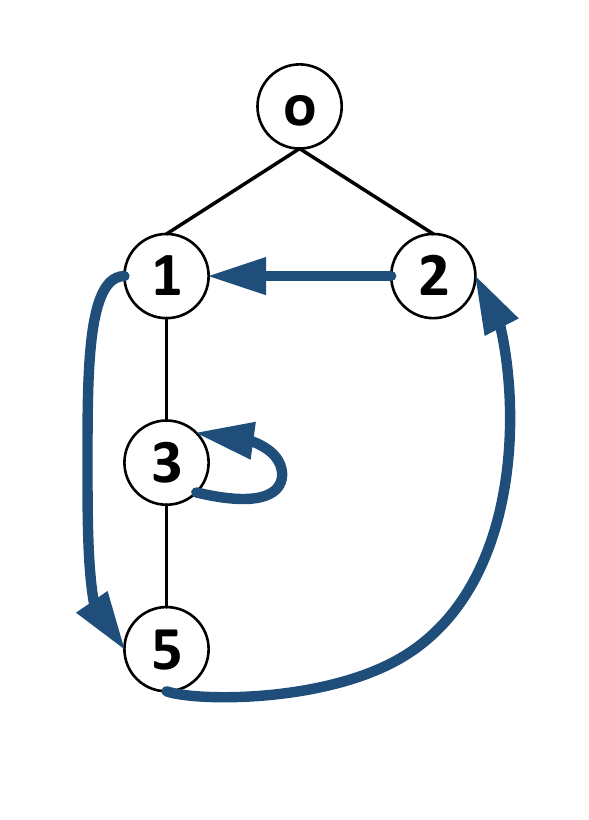}}\\
	\caption{The corresponding graph when agents misreport.}
\end{figure}
The example given in Figure 3 implies a problem that an agent will not invite her neighbors as long as inviting others results in a worse allocation for her. Intuitively, some agent $i$ has to pass some other agents to reach her favorite good in the network, which means $i$ has to ask permissions from these agents. Since some of these agents can determine whether $i$ is in the market or not, they have higher priority to get this good than $i$. In addition, for some cycle in the preference graph, every agent in the cycle has to ask permissions from some other agents outside the cycle with higher priorities. Note that the priority is inherent in the structure of the network and it also relates to agents' preferences. Hence, it is hard to determine whether these agents should agree or not.

Therefore, with the constraints of the network structure and the underlying priority conflicts, we restrict the set of goods that each agent can know about.       

\section{The TTC-Invitation (TTCI) Mechanism}
In this section, we will give our solution to the barter exchange problem in social networks. We begin with a starter mechanism, which will help each agent to get her favorite good by the social network.

\subsection{A Starter}
We have shown the problem when agents know about all the goods in the network. In fact, each agent initially knows the goods of her neighbors. Without the social network, each agent can only exchange with one of her neighbors and she might be refused if her neighbors do not like her good. However, by inviting neighbors, the agent can bring neighbors' neighbors to the trade, then agent might get her favorite good by a larger trade cycle. Inspired by this intuition, we describe a starter mechanism without giving participants any additional information, i.e., they only know the goods of their neighbors. 
\begin{Def}
\rm (\textit{available set}) Given a generated graph $G$, for each $i\in V\setminus \{o\}$, we define the \textit{available set} $\mathcal{G}_i=\{g_j\mid \forall j\in r_i\}$ as the set of goods that $i$ knows. 
\end{Def}

Now let us consider a starter mechanism, each agent reports their preference function. Then $\succ_i=\mathcal{P}_i'(\mathcal{G}_i)$. Similar to the procedure of the TTC algorithm, agents in any cycle of the preference graph can directly trade and will be removed from the set of unassigned agents. Repeat the process until all the agents are reassigned. 

In this mechanism, given the generated graph $G$, we can easily observe that agents who are in some cycle of the preference graph $G'$, they actually form a cycle in the generated graph. This is because $i$ can only point to $j$ if the edge $e(i,j)\in E$ exists. Therefore, the cycle will not be broken by any other agents outside the cycle, which means they can directly trade. With the intuition that such a simple mechanism guarantees the invitation part of the property of IC, it also inherits the desirable properties that the TTC algorithm satisfies since agents in a cycle can freely trade. 

It is noticeable that the starter mechanism helps agents get their neighbors' goods without sacrificing any properties. In the traditional setting, the participants are fixed, which means the agent cannot get her neighbor's good if she is not in the trade cycle with this good. However, in the starter mechanism, the agent can get the good by inviting neighbors in the mechanism so that they can form a larger trade cycle.     
\begin{figure}[htbp]%
	\centering
	\subfigure[]{
		\label{5a}
		\includegraphics[width=0.3\linewidth]{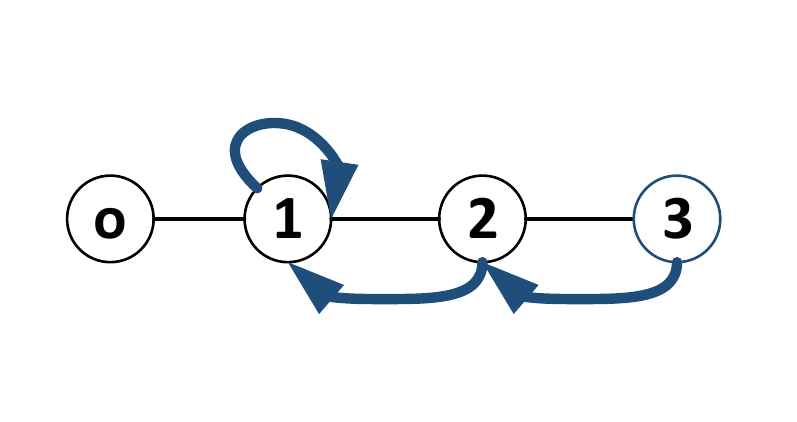}}
	\subfigure[]{%
		\label{5b}%
		\includegraphics[width=0.3\linewidth]{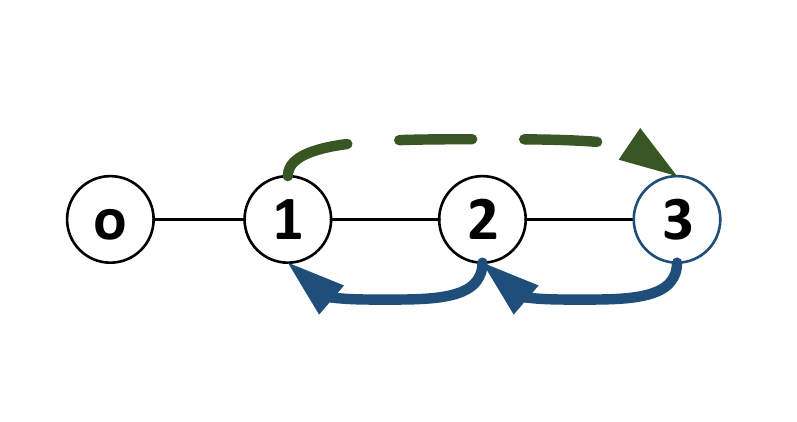}}
	\subfigure[]{%
		\label{5c}%
		\includegraphics[width=0.3\linewidth]{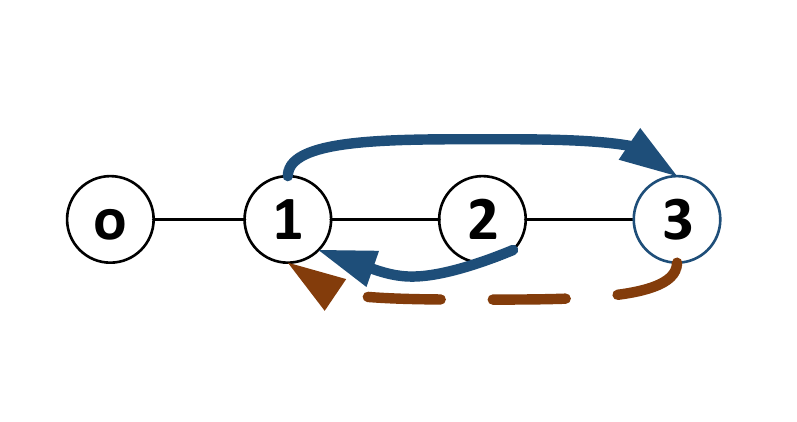}}\\
	\caption{The motivation of our mechanism.}
\end{figure}

Though the starter mechanism guarantees desirable properties, we can still extend the mechanism. Consider the case given in Figure 5, where the preference graph is shown in Figure~\ref{5a}. In Figure~\ref{5a}, agent $3$ points to $2$ but $3$ will prefer $g_1$ to $g_2$ if $3$ knows the existence of $g_1$. Likewise, agent $1$ prefers her own good than $g_2$. If $1$ knows the existence of $3$, $g_3$ will be $1$'s first choice. Inspired by this, we relax the restriction and inform agents about the goods of their descendants. In this case, the preference graph will be shown in Figure~\ref{5b}, where agent $1$ can point to agent $3$ (the directed edge in green). Then we can achieve a better allocation than the starter mechanism does. What if we further relax the restriction and let agents know the goods that their ancestors have as Figure~\ref{5c} shows? In fact, this relaxation will lead to more complex situations, which are similar to the conflicts when directly applying TTC in the network. The case in Figure~\ref{5c} is the same as the case in Figure~\ref{2a}, agent $2$ definitely disapproves of the trade between agent $1$ and $3$ unless she is involved in this cycle. However, if $2$ does not prefer $g_1$ and has already been allocated with her favorite good, the cycle involving agent $1,3$ can definitely trade. The challenge is that it is hard to determine whether $2$ can get her favorite house especially when agents are in a large network. 

Therefore, in our mechanism, each agent only knows the goods from her descendants and neighbors. With the help of the network, agents are more likely to receive goods that they want from their neighbors. The mechanism will inform each agent of their descendants' goods as the additional information given to agents. The additional information can be viewed as bonuses for agents for they propagate the market. 

\subsection{TTC-Invitation (TTCI) Mechanism}
We are now ready to introduce our mechanism. Inspired by the classic TTC algorithm, the mechanism lets agents point to their most favorite remaining and available goods. However, in contrast to TTC, not all the goods are available for agents to pick due to the structural constraint of the network as Figure 2 shows. Intuitively, agent $3$ is not supposed to point to $1$ as she cannot reach $1$ without $2$'s participation. The basic idea of our mechanism is to determine the available set for each agent such that the agent will not be excluded from the trade cycle, which involves both her ancestors and descendants. Here, the available set for each agent is different than the set in the starter mechanism.
\begin{Def}
\rm (\textit{Available Set}) Given a generated graph $G=(V,E)$, for each $i\in V\setminus \{o\}$, we define the available set $\mathcal{G}_i=\{g_j\mid \forall j\in D_i\cup r_i\}$ as the set of goods that $i$ knows. $i$'s true preference order will be $\succ_i=\mathcal{P}_i(\mathcal{G}_i)$. 
\end{Def}
The details of TTCI are given as follows.
\begin{framed}
	\textbf{TTC-Invitation (TTCI) Mechanism}\\
	\noindent\rule{\textwidth}{0.35mm}
	Given the reported type profile $\theta'$, generate the graph $G=(V,E)$. 
    \begin{enumerate}
    \item Initially set $N'= V\setminus \{o\}$.
    \item While $N'\neq \emptyset$:
    \begin{enumerate}
        \item For each $i\in N'$, set $\mathcal{G}_i=\{g_j \mid j\in D_i\cup r_i\}$. The preference order for each $i$ is $\succ_i=\mathcal{P}_i'(\mathcal{G}_i)$.
        \item Construct a preference graph $G'=(V',E')$, where $V'=N'$ and for each $i,j\in N'$, the direct edge $e'(i,j)\in E'$ if and only if $g_j\succ_i g_k$ for all other $g_k\in \mathcal{G}_i\cap N'$. i.e., Each agent $i\in N'$ points to her most favorite remaining and available good (including $g_i$).
        \item Find an arbitrary directed cycle $c=(V_c,E_c)$ in $G'$. Let agents in the cycle trade directly. i.e., for every directed edge $e'(i,j)\in E_c$, set $\mu_{G,\theta'}(i)=g_j$. 
        \item Remove all reassigned agents in the cycle from $N'$: $N'\leftarrow N'\setminus V_c$.
    \end{enumerate}
\end{enumerate}
\end{framed}
Consider the social network given in Figure 6 with seven participants and their preferences are shown in Table 2, the TTCI mechanism runs as follows.
\begin{itemize}
    \item In the first round as shown in Figure~\ref{6a}, $N'=V\setminus \{o\}=\{1,2,3,4,5,6,7\}$.
    \begin{itemize}
        \item Each $i\in N'$ points to her favorite remaining good from $\mathcal{G}_i$.
        \item Find a cycle formed by agent $3,7,5$ and let them trade directly, i.e., agent $3,7,5$ get goods $g_7,g_5,g_3$ respectively.
        \item $N'\leftarrow N'\setminus \{3,5,7\}$.
    \end{itemize}
    \item In the second round (Figure~\ref{6b}), $N'=\{1,2,4,6\}$, and it finds a loop involving $6$. Then $6$ opts out the mechanism with her own good. Then set $N'\leftarrow N'\setminus \{6\}$.
    \item The third round is shown in Figure~\ref{6c}), $N'=\{1,2,4\}$, and it finds a cycle involving $1,4$. Then $1,4$ can directly exchange their goods. Then set $N'\leftarrow N'\setminus \{1,4\}$.
    \item In the last round, as is shown in Figure~\ref{6d}, $N'=\{2\}$. Agent $2$ has no other choice but to leave with her own good.
\end{itemize}
\begin{table}[!ht]
\scalebox{0.7}{
\begin{tabular}{|c|c|c|}
\hline
Agent $i$ & Available Set $\mathcal{G}_i$ & preference order $\succ_i$\\
\hline
1 & $\{g_1,g_2,g_3,g_4,g_5,g_6,g_7\}$ & $g_7\succ_1 g_6\succ_1 g_4\succ_1 g_2\succ_1 g_5\succ_1 g_1\succ_1 g_3$ \\
\hline
2 & $\{g_1,g_2,g_3\}$ & $g_3\succ_2 g_1\succ_2 g_2$ \\
\hline
3 & $\{g_1,g_2,g_3,g_4,g_5,g_6,g_7\}$ & $g_7\succ_3 g_3\succ_3 g_6\succ_3 g_5\succ_3 g_4\succ_3 g_2\succ_3 g_1$ \\
\hline
4 & $\{g_1,g_3,g_4\}$ & $g_3\succ_4 g_1\succ_4 g_4$ \\
\hline
5 & $\{g_3,g_5,g_7\}$ & $g_3\succ_5 g_7\succ_5 g_5$ \\
\hline
6 & $\{g_3,g_6,g_7\}$ & $g_7\succ_6 g_6\succ_6 g_3$ \\
\hline
7 & $\{g_5,g_6,g_7\}$ & $g_5\succ_7 g_6\succ_7 g_7$ \\
\hline
\end{tabular}}
\caption{The preferences of agents in Figure 3.}
\end{table}
\begin{figure}[htbp]%
	\centering
	\subfigure[]{
		\label{6a}
		\includegraphics[width=0.24\linewidth]{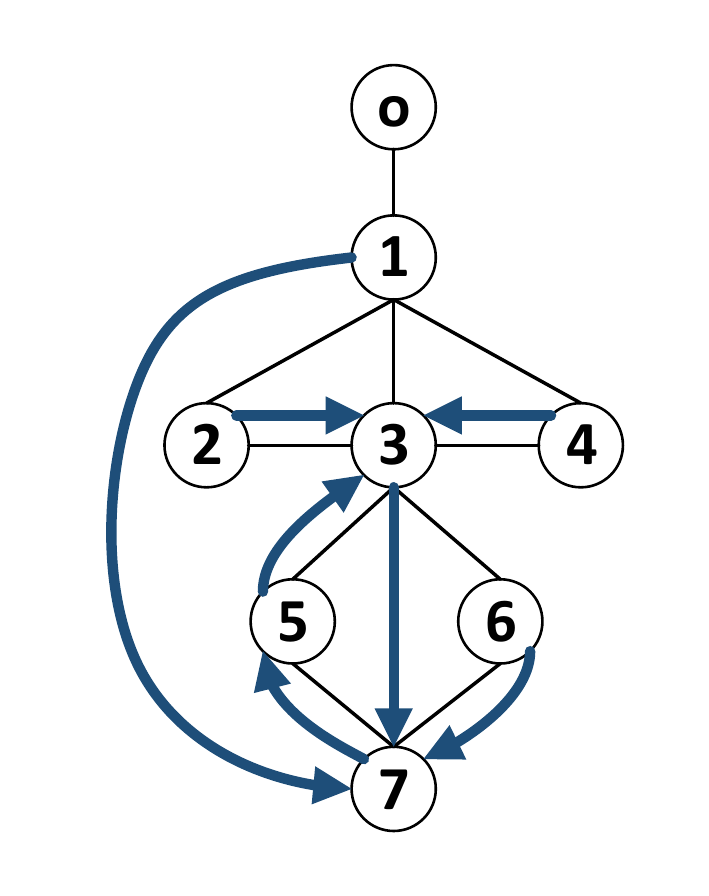}}
	\subfigure[]{%
		\label{6b}%
		\includegraphics[width=0.24\linewidth]{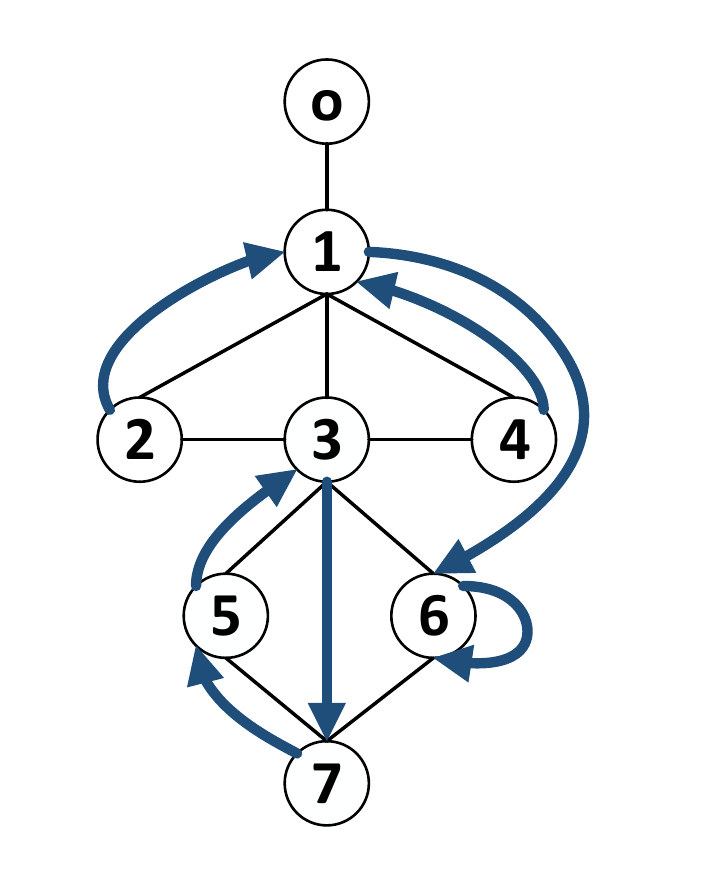}}
	\subfigure[]{%
		\label{6c}%
		\includegraphics[width=0.24\linewidth]{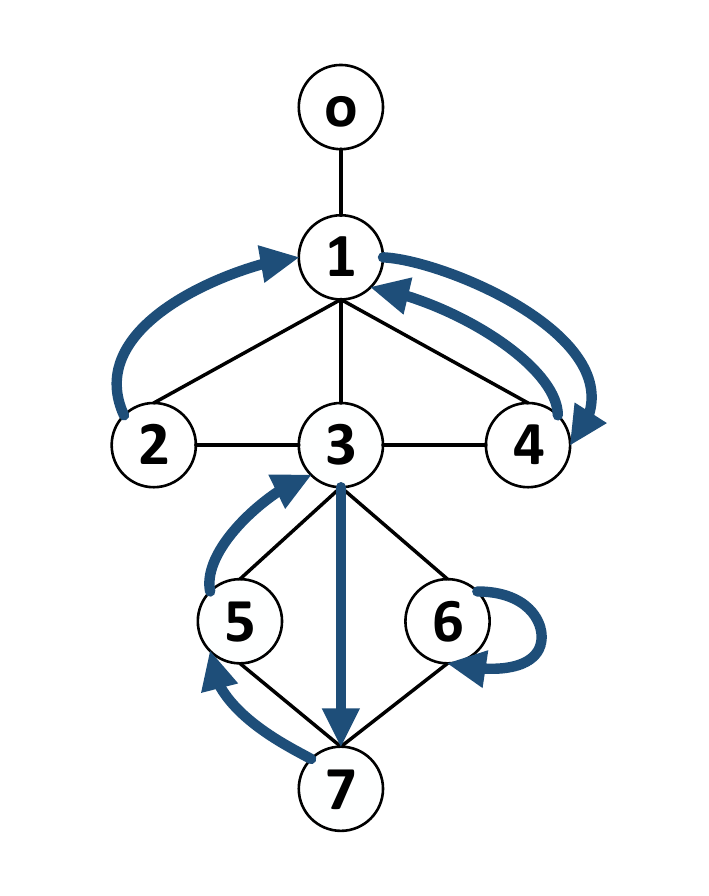}}
	\subfigure[]{%
	\label{6d}%
	\includegraphics[width=0.24\linewidth]{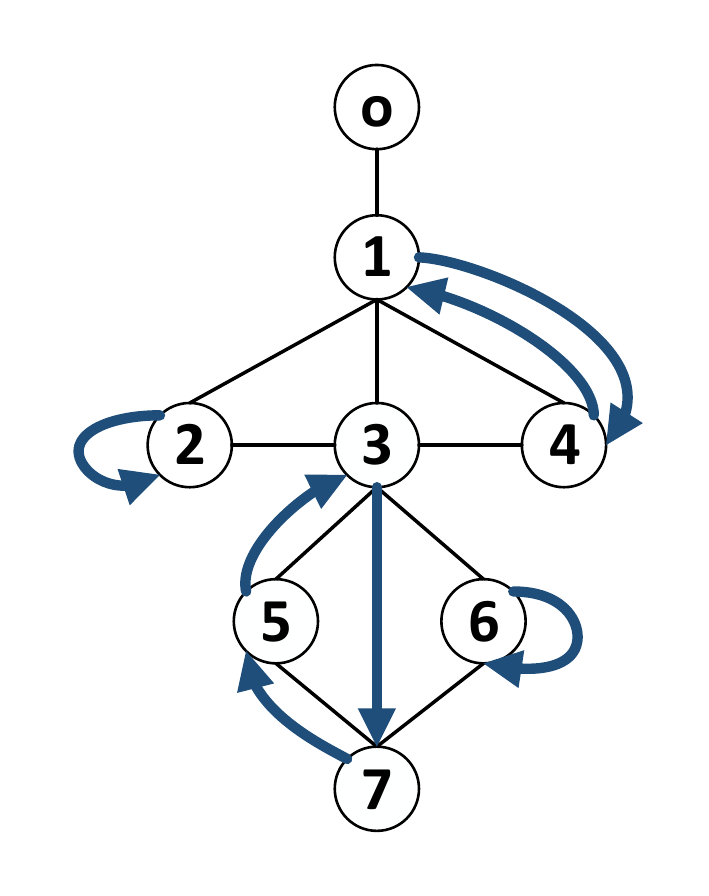}}\\
	\caption{An running example of TTCI.}
\end{figure}
Note that when the generated graph is a completed graph, that is each pair of nodes in the graph is connected by an edge, our mechanism is equivalent to the traditional TTC algorithm. This is because for each agent $i\in V\setminus \{o\}$, $i$ can freely reach any other agent in $V\setminus \{o\}$, i.e., $\mathcal{G}_i=\{g_j\mid \forall j\in V\setminus \{o\}\}$.

\section{Properties of TTCI} 
In this section, we prove that our TTC-Invitation mechanism is individually rational and incentive compatible. our mechanism can attract agents to participate and further enlarge the scale of the market.

\begin{theorem}
\rm The TTC-Invitation (TTCI) mechanism is \textit{individually rational}.
\end{theorem}
We show that by joining the market and honestly reporting the preference function, each agent will get a good that is no worse than the one she can get from her subgraph $G_i$. Here we assume that if the agent refuses to join in the market and she can trade in her subgraph, the mechanism will know this subgraph. Under this assumption, we prove that each agent will not get a worse good by bringing this subgraph to the market. 
\begin{proof}
\rm We assume that agent $i$ reports her true preference function $\mathcal{P}_i$. Suppose $i$ trades in the preference cycle $c_i=(V_{c_i},E_{c_i})$ in $G_i$ and receives $g_j$, i.e., she does not participate the market.

If $i$ joins in the market, she can reach some neighbors who have already been in the market. She probably finds some available goods that are better than $g_j$, then she can point to that good first and might get one in the end. Since for $j$ and any other descendants of $i$, their available sets remain the same after $i$ participating in the market. That is each agent $k\in V_{c_i}\setminus \{i\}$ remains unchanged until $i$ trades in another cycle. In this case, $i$ gets a better good than $g_j$. Even though $i$ might still prefer $g_j$ among all the available goods, or unluckily she fails to get any better good and points back to $g_j$. As mentioned above, all the directed edges in $E_{c_i}$ except $e'(i,j)$ still exist. Thus, once $i$ points to $j$, i.e., the directed edge $e'(i,j)$ exists, the cycle $c_i$ will be formed again and $i$ still gets $g_j$ in the end. 

Therefore, agent $i$ either gets a good better than $g_j$ or still gets $g_j$ when she participates in the market, i.e., TTCI is IR.
\end{proof}

\begin{theorem}
\rm The TTC-Invitation (TTCI) mechanism is \textit{incentive compatible}.
\end{theorem}
To prove our mechanism is IC, we need to show that given the agents' type profile $\theta\in \Theta$ and their reported type profile $\theta'\in \Theta$ and the generated graph $G=(V,E)$, for each agent $i\in V\setminus \{o\}$ such that $\theta_i'\neq nil$:  
\begin{itemize}
    \item Inviting IC: fix $i$'s preference function report to be $\mathcal{P}_i$, inviting all $i$'s neighbours has no negative effect on the allocation to $i$. i.e., $i$ will not get a worse good.
    \item Reporting IC: fix $i$'s invitation to be $r_i'$, $i$ weakly prefers the allocation when she reports her true preference function $\mathcal{P}_i$, over the allocation when $i$ chooses some other preference function.
\end{itemize}
We first show that TTCI satisfies the invitation part of IC.
\begin{proof}
\rm For each $i\neq j\in V\setminus \{o\}$, $j\in r_i\cap D_i$, i.e., agent $i$ can directly contact $j$ and $i$ is $j$'s ancestor. We construct a preference graph $G'$ on $G$ basing on agents' preferences. In the following paragraphs, we will discuss whether $i$ is willing to invite $j$ under two cases. 
\begin{enumerate}
    \item $i$'s favorite good is occupied by the agent in $D_j$. There is no doubt that $i$ invites $j$, otherwise she has no means to get her favorite good.
    \item $i$'s favorite good $g_k\in \mathcal{G}_i$ is in the cycle $c_j=(V_{c_j},E_{c_j})$ which includes $j$. We prove that in this case $i$ cannot be allocated with $g_k$. That is the cycle $c_j$ is irrelevant with $i$'s invitation. There are two possible situations.
    \begin{enumerate}
        \item Both agent $i$ and $j$ prefer good $g_k$.
        \begin{enumerate}
            \item When $j\in V_{c_j}$ and $j$ points to $k$, if $k\in D_i\cap D_j$, then $j$ can reach $k$ without $i$'s invitation. Next we show that $V_{c_j}\subseteq D_j$. It is impossible that the cycle $c_j$ involves agents in $D_k\setminus \{k\}$. Because if so, there should be an agent $x$ in $D_k$ who points to another agent $y$ in $D_j\setminus D_k$. $x$ is not supposed to point to $y$ unless $e(x,y)\in E$. Since there exists a simple path from $y$ to $j$ and to $o$, then $x$ can reach the organizer $\{o\}$ by passing $y$ to the path from $y$ to $o$. That means $x$ can reach $\{o\}$ without passing $k$, thus $k$ is not the ancestor of $x$, a contradiction. Therefore, $e(x,y)\notin E$. Similarly, all agents in $c_j$ belong to $D_j$, which means the existence of $c_j$ is irrelevant with $i$.  
            \item When $k\in D_i$ and $k\notin D_j$, $j$ can point to $k$ indicates that $k\in r_j$. Similar to the previous case, $k$ will not point to her descendant, otherwise there must exist an edge from this descendant to $j$ then to $o$ without passing $k$, which contradicts the assumption. Then each agent in $V_{c_j}$ is pointing to her neighbor and hence the agents in $V_{c_j}$ also form a cycle in $G$. Since $i$ prefers good $g_k$, she will invite $k$. Then every agent in $V_{c_j}$ will invite their neighbors. That is even if $i$ does not invite $j$, agents in $V_{c_j}$ will still be in the market by their neighbors' invitations.
        \end{enumerate}
        \item $j$ does not point to $k$. Similar to the previous proof, the agents in $V_{c_j}$ can form a cycle in $G$. As long as $i$ invites $k$, the cycle $c_j$ will exist. 
    \end{enumerate}
\end{enumerate}
\end{proof}
Next, we show that agents cannot improve their results by misreporting their preference functions $\mathcal{P}_i'$ when their neighbors $r_i'$ are determined. To prove the above, we consider whether the agent can improve her result by pointing to a less preferred good.
\begin{proof}
\rm  We first show that every cycle in $G'$ can directly trade. As proved in the invitation part, when agents who are in a cycle of $G'$ also form a cycle in $G$, then the cycle cannot be broken. On the other hand, if agents in a cycle of $G'$ cannot form a cycle in $G$, then they can reach each other without passing by any other agents outside the cycle. That is one of the agents is in a cycle involving some of her descendant. As the proof in the invitation part has proved, the cycle will not be broken by any other agents outside the cycle, so agents in this cycle can directly change. 

Secondly, we prove that no agent can improve her results by misreporting her preference function. Suppose agent $i$ can get a good $g_j$ at round $t$ by truthful reporting her preference function. Suppose there are $m$ paths either from $i$'s descendant to $i$ or from one of neighbors of $i$ to $i$. That is there are $m$ ``choices" for $i$, and $i$'s allocation is the best among them. At round $t'\ge t$, as long as $i$ is not reassigned, these $m$ paths will remain the same and there might be more choices for $i$, since there are some agents whose favorite goods were traded and they form other paths to $i$. Note that $i$ points to her most favorite remaining good at each round, and she changes only when this good was traded. Thus, the only way that $i$ misreports is that $i$ points to a less preferred good. From this angle, fix other agents' preferences, by misreporting the preference function, if $i$ gets allocated at round $t''\le t$, the number of choices for $i$ might decrease, and the choices at round $t''$ are included within the $m$ choices at round $t$. Therefore, it will not improve the result if $i$ gets allocated before round $t$. On the other side, if $i$ get allocated at round $t'\ge t$, the good $g_k$ will not be better than $g_j$, otherwise $i$ will points to $g_k$ before pointing to $g_j$. Note that it is because all the goods that are better than $g_j$ were traded that $i$ points to $g_j$. Therefore, $i$'s allocation under truthful preference order is the best she can get.
\end{proof}

\begin{theorem}
\rm The TTC-Invitation (TTCI) mechanism outputs a \textit{Pareto optimal} matching $\mu$.
\end{theorem}
\begin{proof}
\rm Suppose $\mu$ is not Pareto optimal. In this case, $\mu$ is Pareto dominated by some other allocation $\nu$. Let us give the proof by induction.
\begin{itemize}
    \item For agents who trade in cycle $c_1=(V_{c_1},E_{c_1})$ at the first round, they must receive the same good in $\nu$ since these agents have received their most favorite good in $\mu$, and $\nu$ cannot allocate them any better goods but the same ones as in $\mu$.
    \item For agents who trade in cycle $c_2=(V_{c_2},E_{c_2})$ at round 2, each $i\in V_{c_2}$ receive their first choice in $\mathcal{G}_i-\{g_j\mid \forall j\in V_{c_1}\}$. Therefore, they cannot get any other better good in $\nu$ but to receive the same in $\nu$, otherwise, they will get worse allocation in $\nu$. 
    \item Similarly, as the above shows, for agents in cycle $c_k=(V_{c_k},E_{c_k})$ at round $k$, we assume that for each $i\in V_{c_1}\cup \cdots \cup V_{c_k}$, $\nu(i)=\mu(i)$. Then for each agent $j$ who trades in $c_{k+1}$ at round $k+1$, we must have $\nu(j)=\mu(j)$. Since for agents trade in $c_{k+1}$, they have received first choice among $\mathcal{G}_j-\{g_j\mid \forall j\in (V_{c_1}\cup \cdots \cup V_{c_k})\}$ and they cannot get any better goods in $\nu$ and thus, $\nu(j)=\mu(j)$. 
    \item Therefore, for any $t\in \mathbb{N}$, we have $\nu=\mu$
\end{itemize}
We have shown that $\nu=\mu$ by induction, which means the assumption does not hold and $\mu$ outputs a \textit{Pareto optimal} matching $\mu$.
\end{proof}

\section{Conclusions and Future Work}
In this paper, we generalize the barter exchange problem in social networks, where each agent owns a good and hopes that more agents can participate in the exchange. We first analyze the problem by directly applying TTC in social networks. Then we prove the impossible result that no mechanism can be IC while also achieves a Pareto optimal matching in the network without any constraints. Following the impossibility result, we restrict the goods that each agent in the network can know about and propose a novel matching mechanism called TTCI. Our mechanism is IC, IR and outputs a Pareto optimal matching. TTCI guarantees that once the agent is in the trade cycle (formed by agents' pointing to their favorite remaining good ) with her favorite good, no other agent in the network can compete with her. This is the key to guarantee that all agents are willing to invite their neighbors even though they have the same preferences.
%without the invitation, the agent can only join the one-to-one exchange with one of her neighbors, and they may fail to exchange. Moreover, without the propagation, the agent can only choose a good from her neighbors', which gives her limited choices. By inviting others, agents can meet more goods to choose but they can also bring competitors. To solve the conflicts, we have proposed a matching mechanism TTCI which incentivizes each agent to truthfully report her preference function and invite all her neighbors. Once the agent is in the trading cycle with her favorite good, no other agent in the network can compete with her. This is the key to guarantee that all agents are willing to invite their neighbors even though they have the same preferences.

One of the most important and challenging future work is to sacrifice the property of Pareto optimal to achieve an IC mechanism where all the goods in the social network are visible to all participants. That is agent can get a good from a totally ``stranger". The main problem here is the priority conflicts. It seems like a Russian Dolls game here. Maybe letting agents point to their ancestors' goods is a good start, or more constraints can be put on the structure of the generated graph.   

\bibliography{mybibliography}
\end{document}